\documentclass[11pt]{article}
\usepackage{graphicx,amssymb,amsmath}

\usepackage[margin=1in]{geometry}

\usepackage[english]{babel}
\usepackage{url}
\usepackage{appendix}
\usepackage{makecell}

\usepackage{amsmath, amssymb, amsthm, amsfonts,alltt,clrscode}
\usepackage{array,colortbl,hhline}
\usepackage[table]{xcolor}
\usepackage{graphicx}
\usepackage[labelfont=bf,font=small]{caption}
\usepackage{todonotes}
\usepackage{float}
\usepackage[list=true,listformat=simple]{subcaption}

\newcolumntype{L}[1]{>{\raggedright\let\newline\\\arraybackslash\hspace{0pt}}m{#1}}
\newcolumntype{C}[1]{>{\centering\let\newline\\\arraybackslash\hspace{0pt}}m{#1}}

\usepackage{hyperref}

\newtheorem{theorem}{Theorem}[section]
\newtheorem{lemma}[theorem]{Lemma}
\newtheorem{corollary}[theorem]{Corollary}

\newtheorem{question}{Question}

\theoremstyle{definition}

\def\emph#1{\textbf{\textit{\boldmath #1}}}

\newenvironment{GAMES}
  {\begin{center}\rowcolors{2}{gray!75}{white}\begin{tabular}}
  {\end{tabular}\end{center}}
\def\SETTING#1{{\small [#1]}}

{\makeatletter
 \gdef\xxxmark{%
   \expandafter\ifx\csname @mpargs\endcsname\relax 
     \expandafter\ifx\csname @captype\endcsname\relax 
       \marginpar{xxx}
     \else
       xxx 
     \fi
   \else
     xxx 
   \fi}
 \gdef\xxx{\@ifnextchar[\xxx@lab\xxx@nolab}
 \long\gdef\xxx@lab[#1]#2{\textbf{[\xxxmark #2 ---{\sc #1}]}}
 \long\gdef\xxx@nolab#1{\textbf{[\xxxmark #1]}}
}

\newif\ifabstract
\abstractfalse
\newif\iffull
\ifabstract \fullfalse \else \fulltrue \fi

\ifabstract

\else

\fi

\newcounter{section-preserve}
\newcounter{theorem-preserve}
\newcommand{\blank}[1]{}
\newtoks\magicAppendix
\magicAppendix={}
\newtoks\magictoks
\newif\iflater
\laterfalse
\long\def\later#1{\iflater#1\else\magictoks={#1}%
	\edef\magictodo{\noexpand\magicAppendix={\the\magicAppendix \par
			\the\magictoks%
	}}
	\magictodo\fi}
\long\def\both#1{\iflater#1\else\magictoks={#1}%
	\edef\magictodo{\noexpand\magicAppendix={\the\magicAppendix \par
			\noexpand\setcounter{theorem-preserve}{\noexpand\arabic{theorem}}%
			\noexpand\setcounter{theorem}{\arabic{theorem}}%
			\noexpand\setcounter{section-preserve}{\noexpand\arabic{section}}%
			\noexpand\setcounter{section}{\arabic{section}}%
			\noexpand\let\noexpand\oldsection=\noexpand\thesection
			\noexpand\def\noexpand\thesection{\thesection}
			\noexpand\let\noexpand\oldlabel=\noexpand\label
			\noexpand\let\noexpand\label=\noexpand\blank
			\the\magictoks%
			\noexpand\setcounter{theorem}{\noexpand\arabic{theorem-preserve}}%
			\noexpand\setcounter{section}{\noexpand\arabic{section-preserve}}%
			\noexpand\let\noexpand\thesection=\noexpand\oldsection
			\noexpand\let\noexpand\label=\noexpand\oldlabel
	}}
	\magictodo
	\the\magictoks\fi}
\def\magicappendix{\latertrue \the\magicAppendix}

\iffull
\long\def\both#1{#1}
\let\later=\both
\def\magicappendix{}
\fi

\title{The Legend of Zelda: The Complexity of Mechanics}

\author{
\begin{tabular}{c@{\hspace{3em}}c@{\hspace{3em}}c}
  Jeffrey Bosboom%
    \thanks{%
      Massachusetts Institute of Technology,
      \protect\url{{jbosboom,brunnerj,mcoulomb,edemaine,dylanhen}@mit.edu}{}{}}
  &
  Josh Brunner\footnotemark[1]
&
  Michael Coulombe\footnotemark[1]
\\[\medskipamount]
  Erik D. Demaine\footnotemark[1]
&
  Dylan H. Hendrickson\footnotemark[1]
&
  Jayson Lynch%
  \thanks{
  University of Waterloo,
Waterloo, Ontario, Canada,
\protect\url{jayson.lynch@waterloo.ca}}
\\[\medskipamount]
&
  Elle Najt%
    \thanks{%
      University of Wisconsin,
      \protect\url{lnajt@math.wisc.edu}}
\end{tabular}
}

\date{}

\begin{document}

\maketitle

\begin{abstract}
  We analyze some of the many game mechanics available to Link in the
  classic Legend of Zelda series of video games.
  In each case, we prove that the generalized game with that mechanic
  is polynomial, NP-complete, NP-hard and in PSPACE, or PSPACE-complete.
  In the process we give an overview of many of the hardness proof techniques
  developed for video games over the past decade:
  the motion-planning-through-gadgets framework,
  the planar doors framework,
  the doors-and-buttons framework,
  the ``Nintendo'' platform game / SAT framework,
  and the collectible tokens and toll roads / Hamiltonicity framework.
\end{abstract}

\section{Introduction}
\label{sec:intro}

\vspace{-0.5\bigskipamount}

\hbox to \hsize{\hfil ``It's dangerous to go alone!  Take this.''}

\medskip

The Legend of Zelda%
\footnote{All products, company names, brand names, trademarks, and sprites are properties of their respective owners.
Sprites are used here under Fair Use for the educational purpose of illustrating mathematical theorems.}
action--adventure video game series consists of 19 main games
developed by Nintendo (sometimes jointly with Capcom),
starting with the famous 1986 original
which sold over 6.5 million copies \cite{vgsales},
and most recently with Breath of the Wild
which was a launch title for Nintendo Switch
(and is arguably what made the Switch an early success).
In each game, the elf protagonist Link explores a world with enemies and obstacles that can be overcome only by specific collectible items and abilities.
Starting with nothing, Link must successively search for items that unlock new areas with further items, until he reaches and defeats a final boss enemy Ganon.

Across the 35-year history of the series, many different mechanics have been introduced, leading to a varied landscape of computational complexity problems to study: what is the difficulty of completing a generalized Zelda game with specific sets of items, abilities, and obstacles?
Reviewing the two Zelda wikis \cite{fandom,zeldadungeon} and playing the games
ourselves, we have identified over 80 unique items with unique mechanics,
listed in Table~\ref{tbl:items}, throughout the 19 games in the Zelda franchise.
More mechanics could likely be identified from the numerous enemy types
and other game features.

In tribute to the fun and challenge of the Zelda series,
we propose a long-term undertaking where the video-game-complexity community
thoroughly catalogs these mechanics and analyzes which combinations lead to
polynomial vs.\ NP-hard computational problems.
Toward this goal, we analyze in this paper the complexity of several new
combinations of various items, including Hookshot, Switch Hook, Diamond Blocks,
Crystal Switches, Roc's Feather, Pegasus Seeds, Kodongos, Buzz Blobs,
Cane of Pacci, magnetic gloves, Magnesis, Bombs, Bow, Ice Arrows, Water,
Fairies, Magic Armor, Decayed Guardians, Statues, Ancient Orbs, and
colored-tile floor puzzles.

Table~\ref{summary} summarizes our results, along with previously known results
about Legend of Zelda.
The first paper to analyze Legend of Zelda games,
from FUN 2014 \cite{Nintendo_TCS}, showed that
Zelda with push-only blocks is NP-hard;
Zelda with Hookshot, push-and-pull blocks, chests, pits, and tunnels is NP-hard;
Zelda with Small Keys, doors, and ledges is NP-hard;
Zelda with ice and sliding push-only blocks is PSPACE-complete;
and Zelda with buttons, doors, teleporter tiles, pits, and Crystal Switches
that activate raised barriers is PSPACE-complete.
More recent work from FUN 2018 \cite{Toggles_FUN2018} showed that
Zelda with spinners is PSPACE-complete.
Many more items and mechanics remain to be analyzed;
refer to Table~\ref{tbl:items} in Section~\ref{sec:open problems}.

Our new results also serve to highlight different techniques for proving
polynomial/NP algorithms, NP-hardness, and PSPACE-hardness of
video games involving the control of a single agent/avatar.
For algorithms, we apply the powerful technique of shortcutting to enable simple searches for solution paths through seemingly complex dungeons, and fixed-parameter tractability analysis to achieve efficiency as dungeons get larger but the game mechanics stay constant.
One major category is Hamiltonian Path inspired reductions, often simplified with Viglietta's Metatheorem 2 \cite{HardGames12} concerning collectible items and toll roads.
Next is the Nintendo-style SAT reduction from \cite{Nintendo_TCS} which later acted as inspiration for the Turrets Metatheorem from \cite{Portal} and the door-opening gadgets in \cite{Gadgets_ITCS2020}.
For PSPACE-hardness, we use the door-and-button framework of Fori\v{s}ek \cite{Forisek10} and Viglietta \cite{HardGames12}.
Finally, we use the door gadget from \cite{Nintendo_TCS} which, along with the other previous work, inspired the gadgets framework for the complexity of motion-planning problems \cite{Toggles_FUN2018,Gadgets_ITCS2020,Doors_FUN2020} which we also use.
As a secondary goal, we hope that this paper offers a nice sampling of proof techniques showing the hardness infrastructure that have been built up in recent years.

We describe our model of generalized Zelda in Section~\ref{sec:model}.
The paper is then organized into sections roughly corresponding to the complexity classes of our results. 
Section~\ref{sec:P} gives polynomial-time algorithms for hookshot and pots; and switch hook and diamond blocks.
Section~\ref{sec:NP} proves NP-hardness for Zelda with floor puzzles; and hookshots, pots, and keys; and a variety of enemies. We show ways of replacing pots and keys with several other sets of items.
Section~\ref{sec:pspace} proves PSPACE-hardness for magnetic gloves; a more powerful Cane of Pacci (while the original version is in P); the Magnesis rune; minecarts with switchable tracks; and Pedestals with Ancient Orbs or Pressure Places with statures that control doors. 

\definecolor{header}{rgb}{0.29,0,0.51}
\definecolor{gray}{rgb}{0.85,0.85,0.85}
\def\header#1{\centerline{\textcolor{white}{\textbf{#1}}}}
\def\fitheader#1{\textcolor{white}{\textbf{#1}}}
\def\maybe#1{\unskip}

\begin{table}[t]
	\centering
	\small
	\rowcolors{2}{gray!75}{white}
	\begin{tabular}{ | L{5cm} | L{4.8cm} | L{2.9cm} | C{1.7cm} |}
		\rowcolor{header}
		\header{Game Mechanics} & \header{Games with Mechanics} & \header{Result} & \header{Thm} \\ 
		Hookshot, Pots, Pits & ALttP, LA, PH, ALBW & $\in$ P & \ref{thm:hookshot-P} \\ 
		Hookshot, Pots, Pits, Keys & ALttP, LA, PH, ALBW & NP-complete & \ref{thm:hook-pot-key}, \ref{cor:hook-pot-key} \\ 
		Switch Hook, Diamond Blocks, Pits & OoA & $\in$ P & \ref{thm:switchhookP} \\ 
		Crystal Switches, Raised Barriers & ALttP, LA, OoA, PH, ALBW & $\in$ P & \ref{thm:CrystalSwitchP}  \\ 
		Roc's Feather, Pegasus Seeds & OoA, OoS, MM & NP-complete & \ref{thm:RocNP}  \\ 
		Bombs, Renewing Cracked Walls & OoT, MM, OoA, OoS, TWW, TMC, TP, ST, SS \maybe{, TFH}, BotW & NP-complete & \ref{thm:BombNP}\\ 
		Ice Arrows, Water & MM & NP-complete & \ref{Thm:IceArrowNP}\\ 
		Healing Items, Unavoidable Damage Region & ALttP, LA, OoT, MM, OoA, OoS, FS, TWW, FSA, TMC, TP, PH, ST, SS, ALBW, BotW & NP-complete & \ref{thm:FairiesNP} \\ 
		Magic Armor, Unavoidable Damage Region & ALttP, OoT, TWW, TP & NP-complete & \ref{thm:InvincibilityNP} \\ 
		Bow or Bombs, and Crystal Switches for Raised Barriers & ALttP, LA, OoA, TP, PH & NP-complete & \ref{CrystalSwitchNP} \\ 
		Colored-tile floor puzzles & LA, OoA, TMC & NP-complete & \ref{thm:ColorFloorNP} \\ 
		Kodongos, low walls, sword & ALttP & NP-hard & \ref{thm:DodongoNP} \\ 
		Buzz Blobs, Master Sword & ALttP, LA, OoA, OoS, TMC, ALBW, TFH & NP-hard & \ref{thm:BuzblobNP} \\ 
		Decayed Guardians, Bombs & BotW & NP-hard & \ref{thm:GuardianNP} \\ 
		Magnetic gloves, metal orbs, ledges, jump platforms & OoS & PSPACE-complete & \ref{sec:magnet-gloves} \\ 
		Cane of Pacci, ground holes, ledges, tunnels& TMC & FPT in duration \newline PSPACE-complete & \ref{thm:CaneOfPacciFPT} \newline \ref{thm:CaneOfPacciPSPACE} \\ 
		Magnesis Rune, metal platforms & BotW & PSPACE-complete & \ref{thm:MagnesisPSPACE} \\ 
		Statues, Pressure Plates, Doors & ALttP, OoT, MM, OoA, OoS, FS, TWW, FSA, TMC, TP, PH, ST, SS, ALBW & PSPACE-complete & \ref{thm:statue-plates} \\ 
		Ancient Orbs, Pedestals, Doors & BotW & PSPACE-complete & \ref{thm:OrbsPSPACE} \\ 
		Minecarts & OoA, OoS, TMC & PSPACE-complete & \ref{sec:Minecarts} \\ 
		\hline
		Pushable Blocks & Zelda I, LA, OoA, OoS, TMC & NP-complete & \cite{Nintendo_TCS} \\ 
		Pushable/pullable Blocks, hookshot, chests, pits, tunnels & ALttP, LA, OoT, MM, TWW, ALBW & NP-complete & \cite{Nintendo_TCS} \\ 
		Keys, Doors, Ledges & AoL, ALttP, LA, OoT, MM, OoA, OoS, FS, TWW, \maybe{FSA,} TMC, TP, PH, ST, SS, ALBW, BotW \maybe{, TFH} & NP-complete & \cite{Nintendo_TCS} \\ 
		Pushable Blocks, Ice & OoT, MM, OoS, TMC, TP, ST & PSPACE-complete & \cite{Nintendo_TCS} \\ 
		Buttons, Doors, Teleporters, Pits, Crystal Switches & ALttP, ALBW & PSPACE-complete & \cite{Nintendo_TCS} \\ 
		Spinners & OoA, OoS & PSPACE-complete & \cite{Toggles_FUN2018} \\ \hline
	\end{tabular}
	\caption{Summary of complexity results for various game mechanics in
    Legend of Zelda games, along with a list of specific games with those
    mechanics abbreviated according to Table~\ref{tbl:game-info}.
  }
	\label{summary}
\end{table}

\begin{table}[htp]
	\centering
	\small
	\rowcolors{2}{gray!75}{white}
    \begin{tabular}{|c|l|c|c|}
	    \rowcolor{header}
        \fitheader{Release Year} & \fitheader{Game Title or Subtitle} & \fitheader{Abbreviation} & \fitheader{Dimensions} \\
        1986 & The Legend of Zelda & LoZ & {2} \\
        1987 & Zelda II: The Adventure of Link & AoL & {2} \\
        1991 & A Link to the Past & ALttP & {2} \\
        1993 & Link's Awakening & LA & {2} \\
        1998 & Ocarina of Time & OoT & {3} \\
        2000 & Majora's Mask & MM & {3} \\
        2001 & Oracle of Ages & OoA & {2} \\
        2001 & Oracle of Seasons & OoS & {2} \\
        2002 & Four Swords & FS & {2} \\
        2002 & The Wind Waker & TWW & {3} \\
        2004 & Four Swords Adventures & FSA & {2} \\
        2004 & The Minish Cap & TMC & {2} \\
        2006 & Twilight Princess & TP & {3} \\
        2007 & Phantom Hourglass & PH & {2.5} \\
        2009 & Spirit Tracks & ST & {2.5} \\
        2011 & Skyward Sword & SS & {3} \\
        2013 & A Link Between Worlds & ALBW & {2.5} \\
        2015 & Tri Force Heroes & TFH & {2.5} \\
        2017 & Breath of the Wild & BotW & {3} \\
    \hline
    \end{tabular}
    \caption{List of games studied in this paper, with the abbreviations used and the number of dimensions. To avoid repetition, we exclude the title prefix ``The Legend of Zelda:''\ present in all games beyond the first two.}
    \label{tbl:game-info}
\end{table}

\section{Zelda Game Model}
\label{sec:model}

Each game in the Legend of Zelda series implements a unique two- or three-dimensional variant of a common base of gameplay mechanics, which we extract and model for the purpose of writing widely applicable proofs in the remainder of this paper.
These models encompass the 19 Zelda games studied in this paper,
listed in Table~\ref{tbl:game-info}.
For brevity, throughout this paper we use abbreviations for the titles of games in the series, which are listed in the table and are commonly used among players.
The table also includes the 2D or 3D classification for each game.
Four games are categorized as ``2.5D'' because, while they are each implemented and visualized as a 3D polygonal world, the top-down gameplay style and item mechanics in many circumstances more closely fit the 2D model described below.

\subsection{2D}

\emph{Generalized 2D Zelda} is a single-player game in which the player controls an avatar, Link, in a two-dimensional world.
The world contains \emph{dungeons}, each consisting of a network of
\emph{rooms}, which are rectangular grids of square \emph{tiles} that set the stage for free-moving dynamic \emph{objects} (enemies, pots, collectible items, etc.).
The goal of the player is to navigate Link from the designated initial room to the designated final room.
Each tile may contain an \emph{obstacle} (a pit, solid wall, short fence, a door, a chest, grass, water, lava, spikes, etc.)\ or be empty.
Link can collect \emph{items} which are recorded in the \emph{inventory} and can change the actions available to Link (such as being able to shoot arrows) or how other mechanics affect Link (such as taking less damage).

In some cases we will refer to categories of mechanics and give examples of specific instances in various games. For example many games have \emph{pits} such as water, lava, or holes in the ground which cause Link to take damage and return to a prior location on the map. We may also have \emph{unavoidable damage regions} such as spiked floors or blade traps which Link can traverse but will cause damage either on contact with each new tile or over time. For some results we will list these categories or a prototypical example and then list specific instances found in specific games to which it applies.

A \emph{collision mask} is a 2D rectangular bitmap representing the space that an object or obstacle occupies at its location, specified with \emph{pixel} precision. A \emph{collision} occurs when two \emph{collision masks} would overlap after updating an object's position, often either preventing that movement, causing damage, or triggering events. A \emph{sprite} is the visual graphic representing an object or obstacle.

Link's position (as well as other dynamic objects' positions) is represented as a fixed-point number of pixels within the current room, although the position of the sprite and collision mask are rounded to integer pixel coordinates, and
the number of fractional bits in coordinates is taken to be a constant.
Time progresses in discrete \emph{frames}, during which the player sets each input button as pressed or released and then the room's objects are updated.
Four directional buttons allow the player to move Link in eight directions at a speed of $1$ pixel per frame (diagonal unit-speed motion is approximately $1/\sqrt{2}$ pixels in each dimension).
Link will \emph{face} in the cardinal direction he has most recently moved in, which determines the direction of his actions.
Link's collision mask is a box whose size is a quarter of a tile (half in each dimension), preventing movement through the collision of solid obstacles or other objects.

The game takes place in one room at a time: the room containing Link. When Link leaves a room, some of its objects and tiles may have their current state forgotten or saved globally (temporarily or permanently) for the next time Link enters. 
Link himself has persistent state, including a \emph{heart meter}, measuring Link's health in quarter hearts, and an inventory containing \emph{equipment} for attacking enemies and traversing obstacles.

Doors between rooms may be defined as \emph{checkpoints}, and the game stores a record of the most recent checkpoint that Link has passed through.
If Link's heart meter becomes empty, Link returns there with his heart meter refilled to a small number of hearts.
The player can also choose to \emph{save} the game at any point, which creates a record of the state of the game, but Link's saved position is set to the most recent checkpoint (and certain obstacles or objects may be saved differently as well to accommodate this). If the player chooses to \emph{quit} the game at any point, the game resets to the saved state.

\subsection{3D}

In \emph{Generalized 3D Zelda}, rooms are instead three-dimensional spaces bounded by solid \emph{polygons} with fixed-point coordinates, as well as dynamic objects with polygon-defined collision boundaries.
Link and other dynamic objects cannot pass through solid polygons, and are pulled downwards by \emph{gravity}, which means that they will fall if they are not on a polygonal surface and will slide down a sufficiently steep surface.
Long falls will cause Link to experience \emph{fall damage} proportional to the distance fallen beyond a safe threshold.

The player uses a \emph{joypad} to control Link's motion, allowing a choice of a fixed number of angles and magnitudes for Link's velocity in the horizontal plane during each frame.
Link has limited \emph{jump} abilities: running off the top of a cliff, Link will jump rather than fall, and when at the bottom of a short ledge, Link will do a jump to climb up the ledge.
Some items are used by \emph{aiming} in Link's first-person view, where the player uses the joypad to adjust the angle the camera points with a fixed-point precision.

Notably, Generalized 2D Zelda reduces to Generalized 3D Zelda by converting the pixels of 2D collision masks of tile obstacles and dynamic objects into polygonal prisms, using a joypad with only four directions and binary magnitude, and aligning everything on top of one large polygon to counter gravity and avoid jumpable ledges.
Consequently, we describe our hardness results for Generalized 2D Zelda unless the third dimension is necessary.

\section{Polynomial-Time Zelda}\label{sec:P}

This section details polynomial-time algorithms for exploring dungeons given certain restricted sets of items and obstacles. The first pair of results are centered around a short-cutting argument guaranteeing that a solution path through a solvable dungeon can be found that never needs to repeatedly visit the same location, due to the locality of Link's abilities. Next, we consider a common mechanic across the series that has global effects on dungeon traversability, and show that in isolation it is easy to overcome.

\subsection{Hookshot and Switch Hook are in P}

\iftrue
The Hookshot was first introduced in The Legend of Zelda: A Link to the Past, a 2D game.
On use, Link shoots a hook in the direction he is facing, which travels at a fixed velocity until it collides with something (or reaches a maximum distance) and then retracts.
If the hook hits a light-weight object, the object will be carried back to Link, but if the hook hits certain heavy objects, Link will be carried to the collision point.

The hookshot allows Link to collect items or move himself across pits, which are obstacles that usually destroy items and damage Link while teleporting him to where he entered the current room.
A common heavy object to hookshot is a pot. Link can also push pots from one tile to an adjacent empty tile, as well as lift a pot over his head and throw it, which destroys it and may damage enemies it hits.
\else
The Hookshot was first introduced in The Legend of Zelda: A Link to the Past, allowing Link to shoot a hook that can grab and pull light-weight objects like rupees towards Link as well as pull Link towards heavier objects like pots or rocks, especially useful for crossing over pits. Link can also push pots, as well as lift a pot above his head to reveal any hidden items underneath them, but then the pot must be thrown and destroyed for Link to free his hands.
\fi

We consider a dungeon containing rooms with only pots and pits, where Link's only item is the hookshot. Starting from the dungeon entrance, Link must push pots, destroy pots, and hookshot onto pots across pits on the path to the dungeon goal room.

\begin{theorem}\label{thm:hookshot-P}
	Generalized 2D Zelda with the hookshot, pots, and pits is in~P.
\end{theorem}

Throughout this paper, each theorem includes a list of games from
Table~\ref{tbl:game-info} and the relevant mechanics from that game
to which the theorem applies:

\begin{GAMES}{|l|lll|}
\rowcolor{header}
	\fitheader{Applicable Games} & \fitheader{Hookshot} & \fitheader{Pot} & \fitheader{Pit}
  \\ 
	ALttP, LA, ALBW & Hookshot & Pot & Water \\
	PH & Grappling Hook & Barrel & Water \\
  \hline
\end{GAMES}

\begin{proof}
	Given a dungeon, we can construct a directed graph $G$ whose vertices are tiles and whose colored edges indicate the ways Link can visit a tile for the first time: if two tiles are adjacent and neither is a pit, then there are red edges going both directions between them, and if it is possible at tile~$A$ to hookshot an existent pot to land on an empty tile~$B$, then there is a blue edge from $A$~to~$B$.
	
	Let a \emph{platform} be a connected component of vertices joined by red edges.
	Notice that Link can always traverse a red edge (by means of lifting and destroying any pot in his way), but may only traverse a blue edge from platform $p_1$ to $p_2$ if there exists a pot on a specific tile in $p_2$.
	This leads to the following observation:
	
	\begin{lemma}\label{thm:hookshot-skip-platform}
    If there is a solution path, then there is a solution path in which Link visits each platform at most once.
    \end{lemma}

    \begin{figure}[htp]
    \centering
    \includegraphics[scale=1.00]{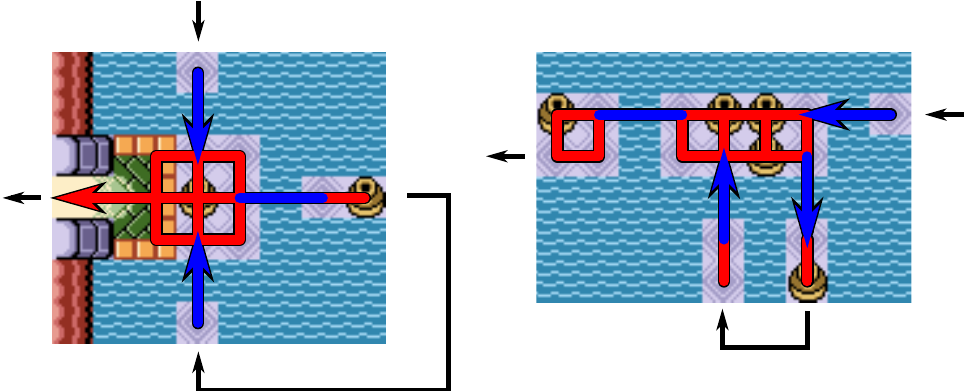}
    \caption{Example cases of Lemma~\ref{thm:hookshot-skip-platform}: (left) If $p$ holds the goal platform, Link can go directly there along red edges. (right) If $p$ isn't the last platform, the last blue edge out of $p$ is already available from red edges when Link first visits $p$.}
    \label{fig:hookshot-shortcut}
    \end{figure}
    
    \begin{proof}
    Refer to Figure~\ref{fig:hookshot-shortcut}.
    Consider the last platform $p$ on a given solution path that Link visits more than once (assuming that $p$ exists, otherwise we are done). We can modify the solution path by skipping from the first time Link visits $p$ to the last visit to $p$.
    
    If the goal tile is on $p$, then once Link enters $p$ for the first time, he may traverse only red edges directly to the goal.
    Otherwise, 
    the path after the last visit to $p$ involves visiting only never-visited-before platforms, so any pots on them were not pushed or destroyed before Link visited $p$. 
    Thus, on Link's first visit to $p$ at tile $t_1$, the blue edge $(t_2,t_3)$ that Link traverses to leave $p$ on his last visit is immediately available, so we can replace the path from $t_1$ to $t_2$ with a path of only red edges in $p$ to get a valid solution that contains strictly fewer platform visits.
    
By well-ordering, we conclude that a shortest solution path cannot have any platforms that Link visits more than once.
    \end{proof}
	
	Using Lemma~\ref{thm:hookshot-skip-platform}, we can derive a polynomial-time algorithm to solve the dungeon. First, we construct the graph $G$, contract its red edges to form a graph $G'$ of platforms with edges corresponding to at least one blue edge in $G$, then run a breadth-first search from the starting tile's platform to the goal tile's platform to find a simple path $q$ of platforms (or determine the dungeon is unsolvable if $q$ does not exist).
	
	Next, we fill in the gaps between $q$'s blue edges to build a simple path $q'$ in $G$. For blue edges $(u,v)$ and $(x,y)$, where tiles $v$ and $x$ are on platform $p$, we run breadth-first search from $v$ to $x$ within the red edges of $p$ and splice the resulting path into $q'$.
	
	Finally, we translate $q'$ into a solution to the dungeon. For each edge $(u,v)$ in $q'$, if the edge is red then we command Link to lift and throw any pot on $v$ then walk from $u$ to $v$, and if the edge is blue then we command Link to hookshot in the direction of $v$.
\end{proof}

A variant of the hookshot was introduced in The Legend of Zelda: Oracle of Ages (also a 2D game), called the Switch Hook. Instead of pulling Link towards a target heavy object, this item swaps their locations, an action which can move unbreakable diamond block obstacles that are otherwise fixed in place.
A similar argument to Lemma~\ref{thm:hookshot-skip-platform}
and Theorem~\ref{thm:hookshot-P} proves the following:

\begin{corollary}\label{thm:switchhookP}
	Generalized 2D Zelda with the switch hook, diamond blocks, and pits is in P.
\end{corollary}

\begin{GAMES}{|l|lll|}
\hline
\rowcolor{header}
	\fitheader{Applicable Games} & \fitheader{Switch Hook} & \fitheader{Diamond Block} & \fitheader{Pit} \\
\hline
	OoA & Switch Hook & Diamond Block & Lava \\
\hline
\end{GAMES}

\begin{proof}
We sketch a proof analogous to Theorem~\ref{thm:hookshot-P}'s proof.
Given a dungeon, we can construct a directed graph $G$ whose vertices are tiles and whose colored edges indicate the ways Link can visit a tile for the first time: if two tiles are adjacent and neither is a pit, then there are red edges going both directions between them, and if it is possible at tile~$A$ to switchhook to an existent diamond block at tile~$B$, where $A$ and $B$ are not adjacent tiles, then there is a blue edge from $A$~to~$B$. See Figure~\ref{fig:switchhook-shortcut} for an example.

\begin{figure}[htp]
\centering
\includegraphics[scale=1.00]{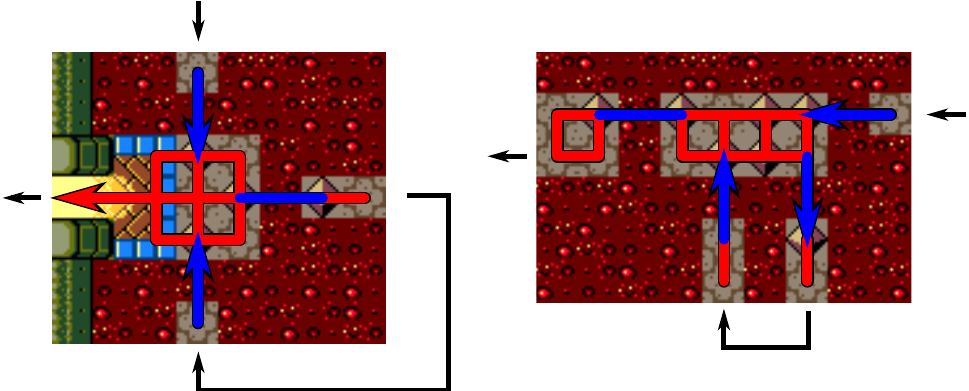}
\caption{Example cases of Corollary~\ref{thm:switchhookP}, analogous to Figure~\ref{fig:hookshot-shortcut}.}
\label{fig:switchhook-shortcut}
\end{figure}

Like Lemma~\ref{thm:hookshot-skip-platform}, here we also get that if there is a solution path, then there is a solution path in which Link visits each platform at most once. Within a platform, Link can reach any tile along red edges by walking through empty tiles and swapping with any adjacent diamond blocks that otherwise block the path. Any unvisited platform must have its diamond blocks in their initial state, since they can only be interacted with using the switchhook, which switches Link's position onto their platform. This means the same shortcutting argument applies to show platform
revisits are unnecessary.

Therefore, constructing the graph $G'$ with all red edges contracted and running a breadth-first search like was described in Theorem~\ref{thm:hookshot-P}'s proof will determine whether a solution path exists in polynomial time.
\end{proof}

\subsection{Crystal Switches With Barriers is in P}
In many Zelda games there are raisable barriers controlled by crystal switches. The barriers can either be lowered, allowed Link to freely pass over them, or raised, blocking entry to that tile. These barriers are colored either red or blue. Crystal switches have both red and blue states and hitting them with swords, boomerangs, bombs, or arrows will cause them to switch state. In the blue state, blue blocks are lowered but red ones are raised, and in the red state the red blocks are lowered but the blue ones are raised. Globally there are only two states for the switches to be in, and thus we can search over the whole state space of the game. However, this will be contrasted in Section~\ref{sec:CrystalNP} where the addition of expendable items which can activate the switches will make the problem NP-complete.

\begin{theorem}\label{thm:CrystalSwitchP}
	Generalized 2D Zelda with one-ways, Crystal Switches, raised barriers, and an inexhaustible way to activate such switches is in P.
\end{theorem}

\begin{GAMES}{|L{5cm}|llll|}
\hline
\rowcolor{header}
	\fitheader{Applicable Games} & \fitheader{One-way} & \fitheader{Crystal Switch} & \fitheader{Raised Barrier} & \fitheader{Activator} \\
\hline
	ALttP \SETTING{Tower of Hera}, LA \SETTING{Bottle Grotto}, OoA \SETTING{Crown Dungeon}, PH \SETTING{Temple of Ice}, ALBW \SETTING{Tower of Hera} & Cliff & Crystal Switch & Raised Barrier & Sword \\
\hline
\end{GAMES}

\begin{proof}
	First, construct a traversability graph for the level in each of the two states for the crystal switches. Next, for each of these graphs examine which locations admit an interaction with a crystal switch. For interactions with bounded range, such as the sword or boomerang, this is a constant for each switch. For something of unbounded range this may be linear in the level size. For each location from which an interaction with a crystal switch is possible, connect the corresponding nodes in the traversability graphs. This new graph is at most quadratic in the level size and we can determine reachability by running a standard graph-search algorithm.
\end{proof}

\section{NP-Hard Zelda}\label{sec:NP}
This section gives NP-hardness results for various mechanics in Zelda. The first set of results uses limited resources needed for traversals to show hardness from Hamiltonian Path in grid graphs. This essentially follows Viglietta's Metatheorem 2 \cite{HardGames12} which shows games containing \emph{collectible cumulative tokens} and \emph{toll roads} which consume these tokens in order to pass them. In Zelda games, one can only carry up to a fixed number of these various items at any given time. We will have to generalize this inventory size for these proofs to apply.

The second set of results details some explicit instances of Hamiltonian Path implemented by the mechanics, and
the third set of results, in Section~\ref{sec:enemiesNP}, uses the ``Nintendo'' platform game NP-hardness framework \cite{Nintendo_TCS} to show various combinations of enemies and weapons in the Zelda series are sufficient for NP-hardness.

\subsection{Collectible Objects}\label{sec:HamPath}

In this section we prove the following theorem:

\begin{theorem}\label{thm:hook-pot-key}
Generalized 2D Zelda with the hookshot, pots, pits, and small keys is NP-hard, if save-and-quit and dying are both prohibited.
\end{theorem}

\begin{GAMES}{|l|llll|}
\hline
\rowcolor{header}
	\fitheader{Applicable Games} & \fitheader{Hookshot} & \fitheader{Pot} & \fitheader{Pit} & \fitheader{Small Key} \\
\hline
	ALttP, LA, ALBW & Hookshot & Pot & Water & Small Key \\
	PH & Grappling Hook & Barrel & Water & Small Key \\
\hline
\end{GAMES}

\begin{proof}
We reduce from the Hamiltonian $s$--$t$ path in maximum-degree-$3$ grid graphs
\cite{Degree3GridHamPath}.
Let $G = (V,E)$ be a maximum-degree-$3$ plane grid graph,
and $s, t \in V$ be two vertices in that graph.
The construction in \cite{Degree3GridHamPath} can easily attain the additional
property that $s$ and $t$ are on the boundary face of $G$.%
\footnote{The reduction is from Hamiltonian cycle, and vertices $s$ and $t$
  (which in fact have degree~$1$) come from a common vertex in a planar graph,
  so they belong to a common face, and the embedding can be chosen so that 
  this face is the outside face.}
We construct a dungeon in the following way; refer to Figure~\ref{fig:hookshot-trim}.
In our construction, we will describe distances such that
the hookshots length is 10 units.

The setting of the dungeon will be platforms surrounded by deep water tiles, so Link starting with only a quarter heart cannot step off of the platforms without dying.
For each vertex $v \in V$ located at $(x,y)$, place a plus shaped tetromino tile centered at $(10x,10y)$. Place a pot containing a key in the center of block of each tetromino. 
Link enters the dungeon at $s$. Following $t$ there is a sequence of $|V|$ doors with an exit at the end. If $|V|$ keys have been collected, then all the doors can be opened.

Link can only move from platform to platform by using the hookshot to move to a platform that has a pot. 
Due to the hookshots length, Link can only move between platforms that are adjacent in the grid graph $G$. 
Since the pot blocks the path from one end of the platform to the other, Link cannot move off the platform unless the pot has been removed. This prevents Link from using the hookshot to reach this platform a second time.
This means that a successful traversal of the dungeon is the same as a Hamiltonian $s$--$t$ path in $G$.
\end{proof}

\begin{figure}[htp]
\centering
\includegraphics[scale=1.00,clip,trim=0 0 0 60px]{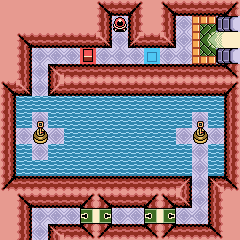}
\caption{Platforms and locked door finale gadgets in the construction of Theorem~\ref{thm:hook-pot-key} from a two-vertex graph.}
\label{fig:hookshot-trim}
\end{figure}

\paragraph*{Complications from Save and Quit.}
One common game mechanic we prohibited above is the ability to save the game, quit to a title screen, and reload the game. In most games in the Legend of Zelda series, Link's inventory is preserved but his location is set to the previous outdoor exit used, and many destructible obstacles like pots are restored to their original state. If Link quits due to dying with zero hearts, then Link may be reset to $3$ hearts.

If we allow the player to die or save and quit, the above construction breaks because that action relocates Link to the beginning of the dungeon and regenerates the pots without removing the collected keys from Link's inventory, which corresponds to allowing the path in the graph to jump to $s$ at any time and reuse edges, violating the Hamiltonian property.

Without prohibiting dying or save and quit, we can show NP-hardness if we can modify the construction to put the game in an unwinnable state if the player uses these mechanics.
This can be done by augmenting the dungeon with a one-use door gadget placed at the entrance.
One way to achieve this is to place an unavoidable damage region (such as floor spikes or flamethrowers) in which traversal of this region will do three hearts of damage to Link. Thus if Link starts the dungeon with more than three hearts, the region will be passable, but upon dying the level will restart with only three hearts and thus the unavoidable damage region will be impassable.

\begin{corollary}\label{cor:hook-pot-key}
Generalized 2D Zelda with the hookshot, pots, pits, and small keys is NP-hard, if save-and-quit and/or dying are allowed.
\end{corollary}

\subsection{Additional Hamiltonian Path Hardness}
\label{apx:more-hampath}

Viglietta's Metatheorem 2 \cite{HardGames12} applies to a broad range of items beyond the hookshot, pots, and keys. We present a collection of other item sets which fit the framework as well.
\begin{itemize}
	\item The Roc's Feather is an item that allows Link to jump a distance of one tile, and Pegasus Seeds are consumable items which temporarily give Link the ability to run faster and jump a distance of two tiles. We require Link to collect seeds to jump over two-tile-wide gaps of lava separated by long-enough distances to wear out their effect.
    
	\begin{corollary}\label{thm:RocNP}
		Generalized 2D Zelda with Roc's Feather, Pegasus Seeds, and lava is NP-complete.
	\end{corollary}
	
\begin{GAMES}{|l|lll|}
    \hline
    \rowcolor{header}
	    \fitheader{Applicable Games} & \fitheader{Roc's Feather} & \fitheader{Pegasus Seed} & \fitheader{Pits} \\
    \hline
	    OoA, OoS & Roc's Feather & Pegasus Seed & Lava \\
        MM & Goron Mask & Magic jar & Pit with jump ramps\\
        TWW & Deku Leaf & Magic jar & Pit with high watchtowers \\
    \hline
\end{GAMES}
	
	\item Explosives, such as bombs, bombchus, and bomb arrows, are consumable items which can destroy certain obstacles, including cracked blocks that are regenerated when Link leaves the room. We require Link to collect explosives to pass through rooms blocked by cracked blocks.
    
	\begin{corollary}\label{thm:BombNP}
		Generalized 2D Zelda with regenerating cracked blocks, explosives, and room transitions is NP-complete.
	\end{corollary}
	
\begin{GAMES}{|L{4.2cm}|L{3.5cm}L{2.5cm}l|}
    \hline
    \rowcolor{header}
	    \fitheader{Applicable Games} & \fitheader{Regenerating Cracked Blocks} & \fitheader{Explosives} & \fitheader{Room Transitions} \\
    \hline
        OoA \SETTING{L3}, OoS \SETTING{L2}, TMC \SETTING{under Hyrule Town} & Cracked Blocks & Bombs & Screen transitions\\
        OoT \SETTING{Goron City} & Brown Boulders & Bombs & Area transitions \\
	    MM \SETTING{Mountain Village} & Snow Boulders & Bombs & Area transitions \\
	    TWW \SETTING{Rock Spire Isle} & Large Cracked Rocks & Bombs & Area transitions \\
	    TP \SETTING{Snowpeak Ruins} & Large Barrels & Bombs & Room transitions \\
	    SS \SETTING{Lanayru Desert} & Rock Piles & Bombs & Area transitions \\
	    BotW \SETTING{Ja Baij Shrine} & Cracked Concrete Cubes & Bomb Arrows, Bow & Area transitions \\
    \hline
\end{GAMES}
	
	\item In The Legend of Zelda: Majora's Mask, Ice Arrows are consumable items that can create temporary ice platforms when shot into water from a Bow with sufficient magic power. We require Link collect arrows and Small Magic Refills to cross pools of water that he cannot climb out of.
    
	\begin{corollary}\label{Thm:IceArrowNP}
		Generalized 3D Zelda with the Bow, Ice Arrows, water pools, and Small Magic Refills is NP-complete.
	\end{corollary}
	
\begin{GAMES}{|l|lll|}
    \hline
    \rowcolor{header}
	    \fitheader{Applicable Games} & \fitheader{Ice Arrows} & \fitheader{Freezable Water} & \fitheader{Small Magic Refill} \\
    \hline
	    MM & Ice arrows & Deep Water & Magic jar \\
    \hline
\end{GAMES}
	
	\item Fairies in bottles are automatic heart-refilling consumable items that refill Link's heart meter when it drops to zero, preventing death. Unavoidable damage regions, such as a hallway with flamethrowers, laser-shooting eyes in the walls, a long fall, or a spiked floor, can be sufficiently large to deal a lethal amount of damage, requiring Link to use one fairy to traverse. While bottles usually occupy limited inventory slots, we consider the case with a number of bottles linear in the dungeon size, so that all required fairies can be carried at once. In Breath of the Wild fairies do not need to be contained in bottles and occupy the inventory like other collectible items.
    
	\begin{corollary}\label{thm:FairiesNP}
		Generalized 2D Zelda with fairies, a linear number of empty bottles, and unavoidable damage regions is NP-complete.
	\end{corollary}
    
\begin{GAMES}{|L{5cm}|ll|}
    \hline
    \rowcolor{header}
	    \fitheader{Applicable Games} & \fitheader{Healing Item} & \fitheader{Unavoidable damage region} \\
    \hline
	    ALttP, ALBW, TMC & Fairy Bottles & Floor spike trap \\
	    OoT, MM, TWW, TP, SS & Fairy Bottles & Long fall \\
	    LA & Secret Medicine & Floor spike trap \\
	    OoA, OoS & Magic Potion & Floor spike trap \\
        PH, ST & Purple/Yellow Potion & Floor spike trap \\
        FS \SETTING{Hero's Trial in Anniversary Edition} & Rupees & Blade trap \\
        FSA \SETTING{Tower of Flames} & Fairies & Fire bars \\
        BotW & Fairies & Malice \\
    \hline
\end{GAMES}
	
	\item Multiple games in the series have magic invincibility items, including the Magic Cape, the Cane of Byrna, Nayru's Love, and the Magic Armor, which consume magic power (or rupees, in The Legend of Zelda: Twilight Princess) to prevent all damage.
	Link can collect Small Magic Refills to increase his magic meter and be required to drain it a specific amount to cross unavoidable damage regions.
    
	\begin{corollary}\label{thm:InvincibilityNP}
		Generalized 2D Zelda with a magic invincibility item, Small Magic Refills, and unavoidable damage regions is NP-complete.
	\end{corollary}

\begin{GAMES}{|l|L{2.5cm}lL{3cm}|}
    \hline
    \rowcolor{header}
	    \fitheader{Applicable Games} & \fitheader{Magic Invincibility} & \fitheader{Magic refill} & \fitheader{Unavoidable damage region} \\
    \hline
	    ALttP & \makecell[l]{Magic Cape or \\ Cane of Byrna} & Magic jar & Floor spike trap \\
	    OoT & Nayru's Love & Magic jar & Blade trap \\
	    TWW, TP & Magic Armor & Magic jar and rupees & Blade trap \\
    \hline
\end{GAMES}

	\item Crystal Switches can be activated by limited use ranged weapons such as bombs or bow and arrows. We can construct a toll road by placing a pair of paths blocked by blocks of both colors. From the center of these obstacles, there is a crystal switch which can be reached by our ranged weapon allowing us to switch the color while in between the barriers and thus traverse them.
	\label{sec:CrystalNP}
    
	\begin{corollary}\label{CrystalSwitchNP}
		Generalized Zelda with Crystal Switches and the bow and arrows or bombs is NP-complete.
	\end{corollary}
    
\begin{GAMES}{|l|lll|}
    \hline
    \rowcolor{header}
	    \fitheader{Applicable Games} & \fitheader{Crystal Switch} & \fitheader{Barriers} & \fitheader{Consumable activator} \\
    \hline
	    ALttP, LA, OoA, PH & Crystal Switch & Barriers & Bombs \\
        TP \SETTING{Temple of Time} & Crystal Switch & Shifting Walls & Arrows\\
    \hline
\end{GAMES}
	
\end{itemize}

\subsection{Floor Puzzles are NP-Hard}

In The Legend of Zelda: Link's Awakening, originally a 2D game, the dungeon Turtle Rock contains puzzles where a flashing block with the ability to replace pits with floor tiles is waiting next to a large set of pits. Link must remotely navigate the block to fill in every pit, which makes a Treasure Chest appear; the challenge being that the block can only traverse over pit tiles.
A similar type of 2D puzzle appears in The Legend of Zelda: Oracle of Ages, where multiple dungeons have rooms with blue floors and a single yellow tile that follows Link's movements. If Link moves onto a blue tile, the yellow tile replaces the blue tile and leaves behind a red tile, and the goal is to eliminate all blue tiles.
    
\begin{theorem}\label{thm:ColorFloorNP}
	Generalized 2D Zelda with pits and flashing floor-generating blocks, and
	Generalized 2D Zelda with colored-tile floor puzzles are both NP-complete.
\end{theorem}

\begin{GAMES}{|l|L{5.5cm}|}
    \hline
    \rowcolor{header}
	    \fitheader{Applicable Games} & \fitheader{Floor Puzzle} \\
    \hline
	    LA \SETTING{Turtle Rock} & Pits and floor-generating block \\
      OoA \SETTING{Skull Dungeon}, TMC \SETTING{Dark Hyrule Castle} & Colored-tiles \\
    \hline
\end{GAMES}

\begin{proof}
	Dungeons with either of these puzzle types are in NP, as a nondeterministic algorithm can guess the buttons to press for Link to traverse the tiles (himself or controlling the flashing block) to solve each room and reach the goal. Without leaving the room, each floor tile can only be filled once by the flashing block and a colored tile can only change from blue to yellow and from yellow to red at most once, therefore there need only be a polynomial number of required moves by monotonicity.
	
	A room containing either of these puzzles is effectively an instance of the Hamiltonian Path problem on a grid graph with pits from a fixed start vertex to any end vertex. We reduce from the problem of Hamiltonian Circuit in grid graphs \cite{GridHamPath} (with no specified endpoints) by laying out the graph using pit tiles (or blue tiles) and replacing one tile on the exterior with the flashing block (or the yellow tile) to specify the starting vertex.
\end{proof}

\subsection{Fighting Monsters is NP-hard}
\label{sec:enemiesNP}
The ``Nintendo'' platform game NP-hardness framework, established in \cite{Nintendo_TCS}, shows several examples of how to use enemies which can be eliminated from one location, but otherwise block another location to build variables and clauses for a SAT reduction. One-way and crossovers are further needed to establish NP-hardness. In \cite{Portal} the notion of what enemies and environments are appropriate is generalized. In particular, we need one pathway which is impossible to cross if the enemy is present (likely because that enemy will kill Link) and another pathway which is disjoint from the first, but allows Link to safely eliminate the enemy. Crossovers and one-ways are prevalent in Zelda games. There are numerous pairing of items and enemies that could be used in this construction; we give a few examples below. In this section we assume enemies do not respawn. This is particularly relevant for The Legend of Zelda: Breath of the Wild where all enemies in the game respawn periodically during the Blood Moon.
    
\begin{theorem}\label{thm:DodongoNP}
	Generalized 2D Zelda with Kodongos, low walls, and a sword is NP-hard.
\end{theorem}

\begin{GAMES}{|l|lll|}
    \hline
    \rowcolor{header}
	    \fitheader{Applicable Games} & \fitheader{Kodongos} & \fitheader{Low wall} & \fitheader{Sword} \\
    \hline
	    ALttP \SETTING{Palace of Darkness} & Kodongos & Low wall & Sword \\
    \hline
\end{GAMES}

\begin{proof}
	Kodongos are enemies which periodically shoot fireballs in a line. For our blocked traversal we will place a Kodongo behind a low wall at the end of a long hallway which is only one tile wide.
	Low walls prevent Link and Kodongos from walking over them, but do allow Link's sword swipes and the Kodongo's fireballs to pass over the wall. 
	 The hallway is long enough that the Kodongo will shoot a fireball down it at least once if Link tries to traverse the hallway while the Kodongo is there. If Link only has half a heart remaining (perhaps from an initial forced traversal of an unavoidable damage region) then this single fireball will kill him.
	
	For our open traversal, we have another long hallway parallel to our blocked traversal but separated by another low wall. This will allow Link to move diagonally adjacent to the Kodongo and safely dispatch it while preventing entry into the other traversal.
\end{proof}

\begin{figure}[htp]
\centering
\includegraphics[scale=1.00]{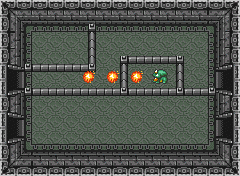}
\caption{Gadget for Theorem~\ref{thm:DodongoNP}}
\label{fig:kodongo-gadget}
\end{figure}
    
\begin{theorem}\label{thm:BuzblobNP}
	Generalized 2D Zelda with Buzz Blobs and the Master Sword is NP-hard.
\end{theorem}

\begin{GAMES}{|l|ll|}
    \hline
    \rowcolor{header}
	    \fitheader{Applicable Games} & \fitheader{Buzz Blobs} & \fitheader{Sword beam ability} \\
    \hline
	    ALttP, OoA, OoS, ALBW & Buzz Blob & Master Sword \\
	    LA & Buzz Blob & Koholint Sword \\
        TMC \SETTING{Palace of Wind} & Electric Chu Chu & Sword Beam or Peril Beam technique \\
        TFH & Buzz Blob & Sword Suit \\
    \hline
\end{GAMES}

\begin{proof}
	Buzz Blobs are enemies which are not injured by sword swipes. If Link attempts to directly hit a Buzz Blob with a sword, Link will take damage instead. The Master Sword allows Link to shoot a ranged attack if he is at full hearts. This ranged attack is able to damage the Buzz Blobs as long as Link is not adjacent to the Buzz Blob.
	
	For our blocked traversal, we will place a Buzz Blob between two ledges. If Link is on top of the first ledge, the ranged attack from the Master Sword will go over the Buzz Blob. If Link jumps down while the Buzz Blob is present, he will take damage. At the very end of the dungeon, we will construct a hallway which is single tile wide filled with Buzz Blobs. If Link is still at full health when reaching this last hallway, Link will be able to dispatch the Buzz Blobs with ranged attacks from the Master Sword. Otherwise it will be impassable.
	
	For the open traversal, we provide a pathway at the same height as the Buzz Blob which is separated by a pit. Link can safely blast the Buzz Blob from across the pit, but cannot cross the pit himself.
\end{proof}

\begin{figure}[htp]
\centering
\includegraphics[scale=1.00]{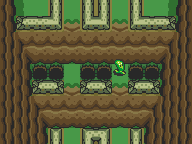}
\caption{Gadget for Theorem~\ref{thm:BuzblobNP}.}
\label{fig:buzz-blob-gadget}
\end{figure}
    
\begin{theorem}\label{thm:GuardianNP}
	Generalized 3D Zelda with Decayed Guardians and bombs is NP-hard.
\end{theorem}

\begin{GAMES}{|l|ll|}
    \hline
    \rowcolor{header}
	    \fitheader{Applicable Games} & \fitheader{Decayed Guardian} & \fitheader{Bombs} \\
    \hline
	    BotW & Decayed Guardian & Bomb Rune \\
    \hline
\end{GAMES}

\begin{proof}
	A Decayed Guardian is an enemy from The Legend of Zelda: Breath of the Wild that has a fixed location but has a laser which can rotate. If Link goes within range the Guardian will take a few seconds to ``lock on'' to Link as a target and then shoot a powerful laser. If Link is unarmored, this deals more than three hearts of damage, the starting maximum number of hearts in the game. However, the Guardian must have a clear line of sight to use its laser.
	
	For our blocked traversal we will have a long, narrow hallway with a Guardian at the end. This hallway will be long enough so that Link will not be able to reach the Guardian before it fires. Further, the hallway is narrow preventing Link from being able to dodge the laser, thus rendering it impassible while the Guardian is there.
	
	For our open traversal we will have another hallway which goes next to the Guardian but is separated by a wall which is slightly taller than Link. This allows Link to safely throw bombs over the wall while not being targeted by the Guardian and not being able to jump over the wall into the blocked traversal. Link could potentially try to bomb-jump over this barrier; however, with only three hearts and no armor attempting a bomb-jump would be lethal. 
\end{proof}

\section{PSPACE-Hard Zelda}\label{sec:pspace}

In this section we show various mechanics in Zelda are sufficient for PSPACE-hardness. 
The first set of results uses the 1-player motion planning framework formalized in \cite{Toggles_FUN2018, Gadgets_ITCS2020} and in particular constructs door gadgets \cite{Mario_FUN2016} and self-closing door gadgets \cite{Doors_FUN2020}.
A \emph{door gadget} has three paths: traverse, open, and close.
The traverse path can be traversed only when the gadget is ``open''.
Traversing the open path opens the gadget, while traversing the close path closes the gadget.
A \emph{self-closing door gadget} is a modified door gadget with two paths, open and traverse, where traversing the traverse pathway also closes it.
An additional set of results, in Section~\ref{sec:switches}, uses the doors-and-buttons framework \cite{HardGames12} and builds 1-switch-2-doors shown PSPACE-hard in \cite{van2015pspace}.

\subsection{Magnetic Gloves is PSPACE-hard}\label{sec:magnet-gloves}

The Magnetic Gloves are an item introduced in The Legend of Zelda: Oracle of Seasons, a 2D game, that projects a north or south magnetic force in any of the four cardinal directions. Among other interactions, they allow Link to remotely attract or repel metal ``N'' orbs, which are polarized north.
Two important properties are the fact that multiple metal objects in range of the force are affected simultaneously, and that metal orbs are affected at any distance, even when off-screen. Since there are no rooms in the game larger than $15\times11$ tiles or containing more than one metal orb, we make the assumptions that the force would affect multiple metal orbs simultaneously and that orbs cannot overlap other orbs, and consider the cases where it has an infinite range and when it has a finite range of at least 15 tiles.
    
\begin{lemma}\label{thm:magnet-pspace}
	Generalized 2D Zelda with magnetic gloves, metal orbs, ledges, and jump platforms is in PSPACE.
\end{lemma}

\begin{GAMES}{|l|llll|}
    \hline
    \rowcolor{header}
	    \fitheader{Applicable Games} & \fitheader{Magnetic gloves} & \fitheader{Metal orbs} & \fitheader{Ledges} & \fitheader{Jump platforms} \\
    \hline
	    OoS & Magnetic gloves & Metal orbs & Ledges & Jump platforms \\
    \hline
\end{GAMES}

\begin{proof}
	Savitch's Theorem \cite{SAVITCH1970177} shows that PSPACE equals NPSPACE, so we give the following simple algorithm to show containment within NPSPACE.
	As rooms are specified by a polynomial number of tiles with pixel-resolution collision masks, 
	the only varying quantities in the game are player-controlled, polynomially bounded, fixed-point positions and velocities of Link and each metal orb, therefore a polynomial-space nondeterministic simulator could guess which buttons to press to find a path to the goal, if a path exists.
\end{proof}

\begin{theorem}\label{thm:magnet-inf}
	Generalized 2D Zelda with infinite-range magnetic gloves, metal orbs, ledges, and jump platforms is PSPACE-hard.
\end{theorem}
\begin{proof}
	We show PSPACE-hardness via reduction from motion planning with door gadgets \cite{Nintendo_TCS}.
	Figure~\ref{fig:magnet-door} shows our construction of a door gadget. In the center of the gadget is a metal orb that always blocks the traverse path (when closed) or the close path (when open). To open the door from the closed state, Link must be in the open path and repel the central metal orb with north magnetic force while facing down. To use the close path while in the open state, Link must use north magnetic force to repel the central metal orb while facing up. If Link tries to attract the central metal orb with south magnetic force, then one of the two ledge orbs will fall and permanently block the traverse path.
	
	In an effort to embed the graph into a single room, we must prevent Link from using the magnetic gloves to manipulate a metal orb inside a gadget from far away. This is solved by entirely surrounding the room with a path with metal orbs on ledges leading to the goal, as in Figure~\ref{fig:magnet-goal-path}. By selectively removing orbs (that would otherwise be dropped to block this path) in rows or columns which we intend the magnetic gloves to be used with a certain polarity, and placing our gadgets on disjoint sets of rows and columns, any unintended magnetic manipulations will permanently block the outer path and prevent the goal from being reached.
\end{proof}
	
\begin{figure}[t]
\centering
\includegraphics[scale=1.00]{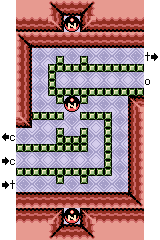}\hspace{4em}\includegraphics[scale=1.00]{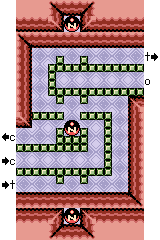}
\caption{Construction of a door gadget using metal orbs, in the closed (left) and open (right) configuration. The open, traverse, and close paths are marked with directions.}
\label{fig:magnet-door}
\end{figure}

\begin{figure}[t]
\centering
\includegraphics[scale=1.00]{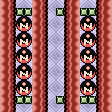}
\hspace{4em}
\includegraphics[scale=1.00]{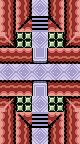}
\caption{(left) Path lined with metal orbs to prevent Link from using the magnetic gloves while facing perpendicular into the path. 
(right) Crossover using jump platforms.}
\label{fig:magnet-goal-path}
\end{figure}

\begin{theorem}
	Generalized 2D Zelda with at least 15-tile range magnetic gloves, metal orbs, ledges, and jump platforms is PSPACE-hard.
\end{theorem}
\begin{proof}
	Compared to infinite range, having a maximum force distance permits black-box gadget constructions, as we prevent external interference by laying-out gadgets far apart in the dungeon.
	However, the construction in Theorem~\ref{thm:magnet-inf} is not self-sufficient because we protected the central metal orb from the left or right by using a single, distant hallway with orbs poised to block traversal to the goal.
	
	We bring these two aspects together by compacting the door gadget enough to run blocking hallways on both sides, as shown in Figure~\ref{fig:magnet-door-finite}. With this 11-tile-wide construction, the metal orbs above the hallways can be placed within the 15-tile limit to protect against horizontal magnetic glove usage on the central metal orb.
	Rather than running the goal hallway around the outside of the room, we thread it past every gadget on both sides, completing the reduction.
\end{proof}
	
\begin{figure}[htp]
\centering
\includegraphics[scale=1.00]{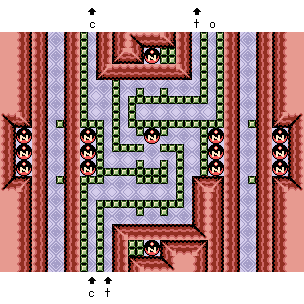}
\caption{Compact construction of a door for 15-tile range magnetic gloves, in the closed state. Hallways on the left and right are traversed at the end to reach the goal.}
\label{fig:magnet-door-finite}
\end{figure}

\subsection{Cane of Pacci is PSPACE-hard}

The Cane of Pacci is an item introduced in The Legend of Zelda: The Minish Cap, a 2D game, that shoots a bolt of magic that can enchant a circular hole tile, which will launch Link up an adjacent ledge if he enters the hole.
As a pseudo-3D effect, the bolt ignores hole tiles that are not ``vertically aligned'' with Link's feet: if the bolt travels down a ledge, then the bolt will remember that it is now high above the floor. The bolt also ignores already-enchanted holes. In the game, the hole stays enchanted for a significant but limited time, so we consider both the finite- and infinite-duration generalizations.
    
\begin{theorem}\label{thm:CaneOfPacciFPT}
	Generalized 2D Zelda with fixed-duration Cane of Pacci, ground holes, ledges, and tunnels is fixed-parameter tractable with respect to cane duration.
\end{theorem}

\begin{GAMES}{|l|llll|}
    \hline
    \rowcolor{header}
	    \fitheader{Applicable Games} & \fitheader{Cane of Pacci (fixed-dur.)} & \fitheader{Ground holes} & \fitheader{Ledges} & \fitheader{Tunnels} \\
    \hline
	    TMC & Cane of Pacci (fixed-dur.) & Ground holes & Ledges & Tunnels \\
    \hline
\end{GAMES}

\begin{proof}
	Let the Cane of Pacci enchant holes for $t$ frames before they automatically unenchant, and let Link's running speed be at most $v \leq 1$ tiles per frame, which is slower than the bolt's travel speed $u$.
	
	For Link to use an enchanted hole, he must be within a circle of radius $v t$ tiles centered at the hole from the duration of the enchantment. Symmetrically, all holes that are beyond $v t$ tiles from his location cannot be enchanted and used, so without loss of generality no strategy for beating the dungeon ever has more than $h = O(v^2 t^2) = O(t^2)$ holes that are enchanted at any point. 
	
	Supposing that there are $n$ square tiles in the world and Link moves at a speed of $1$ pixel per frame, he can be at $O(n/v^2)$ possible positions.
	Link can fire at most one bolt per frame, and each bolt that enchants a reachable hole travels for at most $vt/u < t$ frames.
	Under efficient play, where bolts are only ever shot at reachable holes,
	the total number of game configurations would be $O(n/v^2 \times h t \times (t+1)^h) = n \, (t+1)^{O(t^2)}$.
	
	Therefore, we can create a graph in linear time for fixed $t$, where each node is such a configuration of enchanted holes and to-be-enchanted holes around Link's location, connected by edges representing the effects of possible player inputs on the next frame: Link moving, Link shooting a bolt at a hole in view, or a bolt enchanting a hole. There will be a strategy to get to the end of the dungeon if and only if this graph has a path from the starting configuration node and an ending configuration node.
\end{proof}
    
\begin{theorem}\label{thm:CaneOfPacciPSPACE}
	Generalized 2D Zelda with infinite-duration Cane of Pacci, ground holes, ledges, and tunnels is PSPACE-complete.
\end{theorem}

\begin{GAMES}{|l|llll|}
    \hline
    \rowcolor{header}
	    \fitheader{Applicable Games} & \fitheader{Cane of Pacci ($\infty$-dur.)} & \fitheader{Ground holes} & \fitheader{Ledges} & \fitheader{Tunnels} \\
    \hline
	    TMC & Cane of Pacci ($\infty$-dur.) & Ground holes & Ledges & Tunnels \\
    \hline
\end{GAMES}

\begin{proof}
	Membership in PSPACE is shown similar to Lemma~\ref{thm:magnet-pspace}: given a dungeon as input, there are only polynomially many, polynomial-sized varying quantities in the game state, so if there is a sequence of button presses which brings Link to the goal, then a nondeterministic algorithm can guess and verify them in polynomial space.
	To show PSPACE-hardness, we reduce from motion planning with self-closing doors \cite{Doors_FUN2020}.
	Figure~\ref{fig:pacci-door} shows our design for a self-closing door gadget. Link opens the door by entering the open path and firing the Cane of Pacci over the stone barrier at the hole below the ledge. When open, Link can later traverse by hopping from hole to hole, and the last hole will launch Link up the ledge, disabling the enchantment and thus closing the door behind him. The walls surrounding the holes, and the fact that the cane's bolt does not travel down to lower height levels when shot from the top of a ledge, prevent Link from opening the door anywhere except the open path.
	Because the enchantment does not have a finite duration, Link may be required to open a door but not return to use the door for an arbitrarily long time.
	
	To lay out the graph of self-closing door gadgets in the game, we can make use of the crossover gadget, also shown in Figure~\ref{fig:pacci-door}, if the graph is not planar. Link can freely travel north or south on the upper level, and another path may run left and right by going down stairs and using a tunnel on the lower level.
\end{proof}

\begin{figure}[htp]
\centering
\includegraphics[scale=1.00]{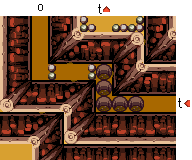} 

$ $

\includegraphics[scale=1.00]{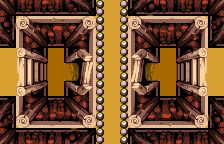}
\caption{Gadgets in The Minish Cap: a self-closing door using holes for the Cane of Pacci (top) and a crossover using tunnels (bottom).}
\label{fig:pacci-door}
\end{figure}

\subsection{Magnesis Rune is PSPACE-Hard}

In The Legend of Zelda: Breath of the Wild, a 3D game, Link obtains the multi-purpose Sheikah Slate, a tool that can be equipped with magical abilities called Runes. Among them is the Magnesis rune, which grants Link telekinetic power over metallic objects within a fixed distance. Compared to the magnetic gloves described in Section~\ref{sec:magnet-gloves}, Magnesis provides full 3D control of exactly one targeted metal object in a world with more-advanced simulated physics, although Link cannot target objects that are out of his line-of-sight or that he is standing on.\footnote{This mechanic intends to prohibit using Magnesis to fly by riding the object being controlled, but there is a glitch that involves stacking multiple specific metal objects to build a ``flying machine'' \cite{flyingmachine}. Our constructions place only a single metal object in a room so that flight cannot be achieved.}
    
\begin{theorem}\label{thm:MagnesisPSPACE}
	Generalized 3D Zelda with the Magnesis rune and large metal plates is PSPACE-hard.
\end{theorem}

\begin{GAMES}{|l|ll|}
    \hline
    \rowcolor{header}
	    \fitheader{Applicable Games} & \fitheader{Magnesis ability} & \fitheader{Large metal plates} \\
    \hline
	    BotW \SETTING{Great Plateau} & Magnesis Rune & Large metal plates \\
    \hline
\end{GAMES}

\begin{proof}
	We reduce from motion planning with self-closing doors \cite{Doors_FUN2020}, using the gadget illustrated in Figure~\ref{fig:magnesis-door}.
	Within a closed room, we construct two paths of platforms over pits: the traverse line, with two gaps that can only be crossed by placing a large metal plate as a bridge, and the open line, raised above the first close enough to use Magnesis on the plate but too far to use it as a bridge to cross paths. Both paths connect to the outside with small exit doors to keep the large metal plate inside.
	
	The self-closing door starts closed, where the large metal plate is not within Magnesis reach of the start of the traverse line.
	To open the door, Link must use Magnesis from the open line to relocate the plate so that when Link later enters the traverse line, he can use the plate as a bridge across both gaps. Carrying the plate from the first gap to the second gap puts it out of Magnesis range of the entrance of the traverse line, which closes the door upon traversal.
\end{proof}
\begin{figure}[htp]
\centering
\includegraphics[scale=1.00]{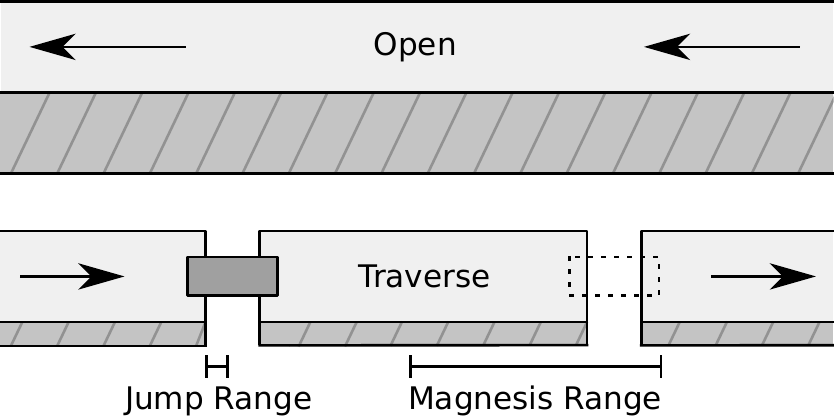}
\caption{Construction of a door gadget using a large metal plate and platforms over pits, shown in the open state. The open line is raised above the traverse line. The layout was inspired by a puzzle in the Oman Au Shrine where the Magnesis rune is unlocked in The Legend of Zelda: Breath of the Wild.}
\label{fig:magnesis-door}
\end{figure}

\subsection{Statues and Pressure Plates are PSPACE-complete}
\label{sec:switches}

Many games in the Legend of Zelda series include dungeon puzzles where, for example, a statue needs to be pushed onto a pressure plate that opens a nearby door as long as the it is pressed down, but the statue can be pushed off of the pressure plate to be used for another use. More recently, in The Legend of Zelda: Breath of the Wild, many shrines contain Ancient Orbs, which can be carried around and placed in Ancient Pedestals to activate other nearby ancient technology.

These statues and orbs act as temporary, reusable keys that exists as objects in the world, as opposed to a collected inventory item like a Small Key that is permanently consumed to open a locked door. As long as we also have barriers to prevent the arbitrary transportation of these objects, this type of puzzle mechanic is PSPACE-hard to solve.

We use the doors-and-buttons framework from \cite{Forisek10, HardGames12}. In these problems the world contains \emph{buttons} which are connected to \emph{doors}. Doors can be open or closed, preventing movement when they are closed. Pressing a button can change whether its doors are open or closed. A version of this, which we call 1-switch-2-doors, in which each button is connected to two doors and pressing the button swaps the openness of both doors was shown to be PSPACE-complete in \cite{van2015pspace}.
    
\begin{theorem}\label{thm:statue-plates}
	Generalized 2D Zelda with pushable heavy statues, pressure plates, doors, and stairs is PSPACE-hard.
\end{theorem}

\begin{GAMES}{|L{5cm}|llll|}
    \hline
    \rowcolor{header}
	    \fitheader{Applicable Games} & \fitheader{Pushable Statue} & \fitheader{Pressure Plate} & \fitheader{Door} & \fitheader{Stairs} \\
    \hline
	    ALttP, OoA, OoS, FS, TWW, FSA, TMC, ALBW & Statue & Floor Button & Door & Stairs \\
	    OoT, MM & Wooden Box & Small Pressure Plate & Door & Stairs \\
        TP \SETTING{Temple of Time} & Pot & Pressure Plate & Door & Stairs \\
        PH, ST, SS \SETTING{Pirate Stronghold} & Metal Box & Floor Button & Door & Stairs \\
	    BotW (see Corollary~\ref{thm:OrbsPSPACE}) & Ancient Orb & Ancient Pedestal & Door & Ladder \\
    \hline
\end{GAMES}

\begin{proof}
	We reduce from motion planning with 1-switch-2-doors gadgets \cite{van2015pspace}. Shown in Figure~\ref{fig:pressure-plates-gadget}, the gadget consists of a tri-partitioned room, one part with one statue and two pressure plates, each controlling a door in the other two parts.
	The statue partition's entrance is raised up from the floor by stairs, preventing Link from pushing the statue outside, and the other two parts each have two entrances on opposite sides of their inner door. Link may choose to open either door by pushing the statue onto the corresponding pressure plate, but with the one statue, at most one door can be opened at a time.
\end{proof}

\begin{corollary}\label{thm:OrbsPSPACE}
	Generalized 3D Zelda with Ancient Orbs, Pedestals, and Doors, along with ladders, is PSPACE-hard.
\end{corollary}

\begin{proof}
	We use the same construction as Theorem~\ref{thm:statue-plates}, replacing statues with Ancient Orbs, pressure plates with Ancient Pedestals, and short steps with ladders. Link is unable to carry an orb while climbing up ladders.
\end{proof}

\begin{figure}[htp]
\centering
\includegraphics[scale=1.00]{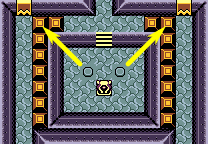}
\caption{Construction for Theorem~\ref{thm:statue-plates}, in Oracle of Ages. The floor buttons open the corresponding shutter doors (signified with arrows) when the statue is pushed on them.}
\label{fig:pressure-plates-gadget}
\end{figure}

\subsection{Minecarts Navigation} \label{sec:Minecarts}

Minecarts are an environmental feature appearing in a variety of Zelda games which pose unique navigational challenges.
In these games, Link can ride minecarts along paths of minecart tracks fixed onto the ground (or raised in the air), ending at minecart stops.
Tracks may also pass through special doors which only open to let a minecart through.
Levers can be used to change the state of sections of track, which can open new paths or create dead-ends that reflect the minecart backwards.
We give PSPACE-completeness proofs that address the types of minecarts found in The Minish Cap, Oracle of Ages, and Oracle of Seasons.

In The Legend of Zelda: Oracle of Ages and Seasons, Link can ride on minecarts which automatically transport him slowly along a track to a destination with no control during the ride beyond the use of some items, such as the sword or bombs. There are also levers Link can switch to rotate certain sections of track $90^\circ$ to toggle the available path through a T-junction. In The Legend of Zelda: The Minish Cap, the dungeon Cave of Flames has fast minecarts that act similarly, although no items may be used during transport, and there are also four way junctions which can be switched between connecting opposite pairs of tracks. Falling off the end of the track or crashing into other minecarts is not a possible situation in the setups in any of these games, so we avoid that situation in our proofs. However, an intermediary simplyfying step considers a model where minecarts bounce off of each other when they collide.

\begin{figure}[htp]
\centering
\includegraphics[scale=1.00]{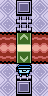} \hspace{4em} 
\includegraphics[scale=1.00]{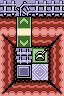} \hspace{4em} 
\includegraphics[scale=1.00]{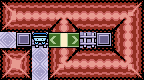}
\caption{Gadgets for Oracle of Ages/Seasons: 1-Toggle while walking (left), 1-Toggle while riding a minecart (center), Diode while walking (right) }
\label{fig:minecart-basic-gadgets}
\end{figure}

\begin{figure}[htp]
\centering
\includegraphics[scale=1.00]{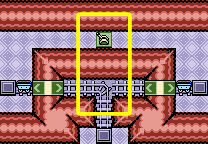}
\caption{Minecart 2-to-1 Toggle for Oracle of Ages/Seasons, simplified under the assumption that minecarts bounce off of stationary minecarts. Adding minecart 1-toggles at each stop would achieve the same bouncing effect.}
\label{fig:minecart-2-to-1}
\end{figure}
    
\begin{theorem}
	Generalized 2D Zelda with a sword, minecarts with tracks, levers to switch T-junctions, and minecart-only doors is PSPACE-complete.
\end{theorem}

\begin{GAMES}{|l|llll|}
    \hline
    \rowcolor{header}
	    \fitheader{Applicable Games} & \fitheader{Minecarts} & \fitheader{Switches} & \fitheader{Blocking T-junction} & \fitheader{Minecart-only door} \\
    \hline
	    OoA, OoS, TMC & Minecarts & Levers & T-junction tile & Minecart-only door \\
    \hline
\end{GAMES}

\begin{proof}
	We reduce from motion-planning with Locking 2-Toggles \cite{Gadgets_ITCS2020}.
	
	We start with basic gadgets, shown in Figure~\ref{fig:minecart-basic-gadgets}. The left image shows a 1-toggle while walking, which consists of a single minecart which can through a minecart-only door. When a minecart is present on Link's side of the door, he can ride it to the other side, and otherwise Link can't do anything else. The 1-toggle while riding a minecart consists of a single T-junction which leads to a minecart stop with a lever controlling the junction. When Link is riding a minecart toward the junction, if it is rotated toward him, he will end up in the minecart stop. At this point, he can flip the lever, and get back in the minecart to continue out the other side. If he enters the junction when it is rotated away from him, he bounces off and returns to where he came from. Although the construction for a diode is not used in our proof we present it here because it may be useful in future constructions and it shows that Zelda with minecarts is not reversible, which initially surprised us.
	
	Figure~\ref{fig:oracle-minecart-l2t} shows our construction of a locking 2-toggle. Without loss of generality, we assume that if a moving minecart collides with a stationary minecart at a minecart stop, it will bounce backwards in the same way that it will bounce off of the dead-end side of a junction. The top half of Figure~\ref{fig:oracle-minecart-l2t} depicts the simplified construction using this assumption. By adding minecart 1-toggles in front of every minecart stop, as we also show in the bottom half of Figure~\ref{fig:oracle-minecart-l2t}, this simplifying assumption can be dropped.
	
	The construction is centered around the 2-to-1 toggle gadget displayed in Figure~\ref{fig:minecart-2-to-1} under our simplifying assumption that minecarts which collide simply bounce off of each other and return the direction they came from. This initially allows Link to enter from the left or the right side and exit from the bottom. Afterwards, Link is only able to go from the bottom to the side from which he last came. Thus this looks like a locking 2-toggle with two of the locations merged.
	
	Now consider the entire gadget which has four entrances made of 1-toggles in the four corners of the gadget. To traverse from the bottom-left to the top-left, first Link uses the 1-toggle to enter the left section. The corridor to the top-left exit is blocked by a minecart, and the only way to move it is to ride it to the central minecart stop and return through the bottom 1-toggle. If the central junction was set to turn right, then Link must first flip the lever, which is easily accessible from both the left and right sections.
	
	In the open state, the central minecart stop is unoccupied, so Link can successfully relocate the cart blocking his path and proceed by using the 1-toggle at the top-left exit. If Link tries later to use the right traversal line, he will not be able to relocate the minecart blocking the top-right exit because the central minecart stop will be occupied. To undo this traversal and restore the open state, Link just needs to retrace his steps, and because the minecart brought from the central area blocks the bottom-left entrance, that is Link's only available choice.
	
	By symmetry, the above arguments also apply to right-line traversals.
\end{proof}

\begin{figure}[htp!]
\centering
\includegraphics[width=0.7\textwidth]{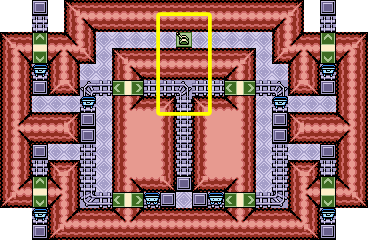}
\includegraphics[width=\textwidth]{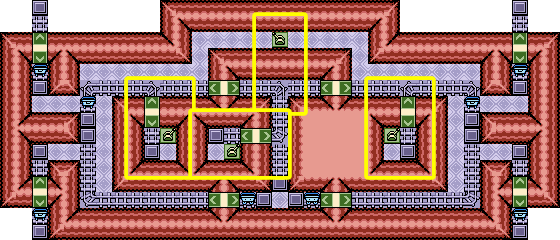}
\caption{Minecart-based Locking 2-Toggle gadget for Oracle of Ages/Seasons. The simplified top figure assumes minecarts bounce off of stationary minecarts, while the bottom adds minecart 1-toggles to get the same effect. Levers and the junctions they switch are highlighted in yellow. Shown in the open state. The traversal lines go from bottom to top on the left and on the right.}
\label{fig:oracle-minecart-l2t}
\end{figure}

The above proof made clever use of the fact that minecarts can block Link's path, forcing link to enter the minecart. It would be interesting to know whether this is necessary for hardness.

\begin{question}
	Can one show PSPACE-hardness for Zelda minecarts with T-junctions without using the fact that minecarts can be used to block paths?
\end{question}

It may also be of interest, or enjoyable to find a simpler reduction for the four way junction. It is tempting to say hardness should follow from the 1-switch-2-doors result of \cite{van2015pspace}, however, the minecarts create a 1-toggle like constraint on the pathways. This suggests a reduction from a toggle-lock or locking 2-toggle will be more appropriate.

\section{Open Problems}
\label{sec:open problems}

Table~\ref{tbl:items} lists many of the items, mechanics, and obstacle types
from across all of the Zelda games, along with known complexity results.
This table gives a sense of the significant work left to complete the
quest of Zelda complexity.

\begin{table}[btp]
	\centering
	\small
	\rowcolors{2}{gray!75}{white}
	\def\maybe#1{\unskip}
	\begin{tabular}{ | L{5.2cm} | L{2.4cm} | }
		\rowcolor{header}
		\header{Items and Mechanics} & \header{Known} \\ 
		Small Key & Theorem \ref{thm:hook-pot-key} \\
        Sword & Theorem \ref{thm:DodongoNP} \\
        Bow, Slingshot, Seed Shooter & Theorem \ref{CrystalSwitchNP} \\
        Shield, Mirror Shield & \\
        Bombs & Theorems \ref{thm:BombNP}, \ref{CrystalSwitchNP}, \ref{thm:GuardianNP} \\
        Boomerang & \\
        Heart Container & \\
        Fairy, Secret Medicine & Theorem \ref{thm:FairiesNP} \\
        Health Refill Potion & \\
        Flippers, Zora Armor Set & \\
        Piece of Heart, Spirit Orbs & \\
        Hookshot, Grapple Hook, Clawshot, Gripshot & Theorems \ref{thm:hookshot-P}, \ref{thm:hook-pot-key}, Prev Work \cite{Nintendo_TCS} \\
        Power Bracelet & \\
        Sword Beams & Theorem \ref{thm:BuzblobNP} \\
        Warp Song, Warp Seeds, Warp Bell & \\
        Blue Ring, Blue Mail & \\
        Hammer & \\
        Pegasus Boots, Pegasus Seeds & Theorem \ref{thm:RocNP} \\
        Shovel, Digging Mitts, Mole Mitts & \\
        Magic Rod, Fire Rod, Fire Gloves & \\
        Roc's Feather, Roc's Cape & Theorem \ref{thm:RocNP} \\
        Candle, Lamp & \\
        Fire Arrows, Ember Seeds & \\
        Gust Jar, Deku Leaf, Whirlwind, Gust Bellows & \\
        Magic Refill Potion & \\
        Bombchu & \\
        Four Sword, Dominion Rod, Command Melody & \\
        Bombos, Ether, Quake, Din's Fire & \\
        Remote Boomerang & \\
        Deku Stick, Boku Stick & \\
        Fire Resist, Lava Resist, Heat Resist & \\
        Ice Arrows & Theorem \ref{Thm:IceArrowNP} \\
        Cane of Byrna, Magic Cape, Nayru's Love, Magic Armor & Theorem \ref{thm:InvincibilityNP} \\ 
        Ember Seeds, Magic Powder, Oil Lantern & \\
        Red Ring, Red Mail & \\
        Timed Expiring Items & \\
        Bomb Arrows & \\
        Bug Net & \\
        Ice Rod, Cryonis Rune & \\
        Magic Mirror, Harp of Ages, Rod of Seasons & \\
        \hline
	\end{tabular}
    \begin{tabular}{ | L{5cm} | L{2cm} | }
		\rowcolor{header}
		\header{Items and Mechanics} & \header{Known} \\ 
        Magnetic Gloves, Magnesis & Theorems \ref{sec:magnet-gloves}, \ref{thm:MagnesisPSPACE}\\
        Super Bomb, Powder Keg & \\
        Silver Scale, Golden Scale, Zora Tunic, Zora Armor, Mermaid Suit & \\
        Cane of Somaria & \\
        Chain Chomp & \\
        Deku Nuts & \\
        Farore's Wind, Travel Medallion & \\
        Hover Boots & \\
        Iron Boots & \\
        Moon Pearl, Portals, Stumps & \\
        Paraglider, Deku Leaf & Theorem \ref{thm:RocNP} \\
        Remote Bomb & \\
        Sand Wand, Sand Rod & \\
        Tornado Rod & \\
        Whip & \\
        Air Potion & \\
        Axe & \\
        Ball and Chain & \\
        Beetle, Hook Beetle & \\
        Cane of Pacci & Theorem \ref{thm:CaneOfPacciFPT} \\
        Lightning Rod & \\
        Minish Cap, Gnat Hat & \\
        Phantom Hourglass, Sand of Hours & \\
        Ravio's Bracelet & \\
        Song of Time 3-day reset & \\
        Spinner (item) & \\
        Stasis Rune & \\
        Switch Hook & Theorem \ref{thm:switchhookP} \\
        Tingle Tuner & \\
        Water Rod & \\
        \hline
        \hline
		\rowcolor{header}
		\header{Obstacles} & \header{Known} \\ 
        Raised red \& blue barriers & Theorems \ref{thm:CrystalSwitchP}, \ref{CrystalSwitchNP}, Prev Work \cite{Nintendo_TCS} \\
        Pots & Theorems \ref{thm:hookshot-P}, \ref{thm:hook-pot-key} \\
        Pits, unswimmable water or lava & Theorems \ref{thm:hookshot-P}, \ref{thm:switchhookP}, \ref{thm:RocNP}, \ref{thm:ColorFloorNP}, \ref{thm:BuzblobNP}, \ref{thm:MagnesisPSPACE} \\
        Minecarts & Theorem \ref{sec:Minecarts} \\
        Floor tile puzzles & Theorem \ref{thm:ColorFloorNP} \\
        Spinners (obstacle) & Prev Work \cite{Toggles_FUN2018} \\
        Metal 3D Physics Objects & Theorem \ref{thm:MagnesisPSPACE} \\
        Floor spikes, walkable lava or fire, long falls & Theorems \ref{thm:FairiesNP}, \ref{thm:InvincibilityNP} \\ 
		\hline 
	\end{tabular}
	\caption{All Items and Mechanics (and associated known results) from across all the Zelda games, as documented on \cite{zeldadungeon} and \cite{fandom}, plus known results for some Obstacles.}
	\label{tbl:items}
\end{table}

It appears the first two Zelda games, The Legend of Zelda and The Adventure of Link, are the only Zelda games which have not been shown to be PSPACE-complete. Resolving the NP versus PSPACE gap for these oldest examples is a remaining challenge. 

In The Legend of Zelda: Oracle of Ages, the Crown Dungeon has a collection of block-pushing puzzles with an interesting twist: all blocks with the same color move simultaneously (if there is an empty tile to move into) when any one of them is pushed. This global manipulation is similar to a discretization of the uniform global control for swarm robotics studied in \cite{becker2013massive}.

The Iron Boots are a common item in 3D Zelda games, first introduced in The Legend of Zelda: Ocarina of Time to allow Link to sink underwater further than he could swim, and was expanded upon in The Legend of Zelda: The Wind Waker as a means to walk against strong winds and activate springboards, and even further in The Legend of Zelda: Twilight Princess by adding interactions with magnetic forces. While the Iron Boots are a nonconsumable inventory item with no inherent motive force, the fact that Link must choose whether or not to wear them to traverse a variety of hazardous terrain gives them the potential to be useful when combined with other items.

\section*{Acknowledgments}

This work was initiated during open problem solving in the MIT class on
Algorithmic Lower Bounds: Fun with Hardness Proofs (6.892)
taught by Erik Demaine in Spring 2019.
We thank the other participants of that class
for related discussions and providing an inspiring atmosphere.
In particular, we thank Lily Chung for contributing to the polynomial-time algorithms in this paper.

We also thank all the contributors to Zelda Wiki \cite{fandom} and Zelda Dungeon \cite{zeldadungeon} for their invaluable information, and to The Spriters Resource \cite{mistermike} and VGMaps.com \cite{rocktyt} for serving as indispensable tools for providing easy and comprehensive access to the sprites used in our figures.

Finally, of course, we thank Nintendo, Capcom, and other associated developers for bringing these timeless classics to the world.

\newpage

\bibliographystyle{alpha}
\bibliography{bibliography}

\appendix

\magicappendix


\end{document}